\PassOptionsToPackage{hyphens}{url}
\documentclass[a4paper,USenglish,nolineno]{socg-lipics-v2021}

\pdfoutput=1 %
\hideLIPIcs  %

\title{Topologically Stable Hough Transform} %

\author{Stefan Huber}{Josef Ressel Centre for Intelligent and Secure Industrial Automation,\\ Salzburg University of Applied Sciences, Austria \and \url{https://www.sthu.org/}}{stefan.huber@fh-salzburg.ac.at}{https://orcid.org/0000-0002-8871-5814}{}%

\author{Kristóf Huszár}{Institute of Geometry, Graz University of Technology, Austria \and \url{https://kristofhuszar.github.io/}}{kristof.huszar@tugraz.at}{https://orcid.org/0000-0002-5445-5057}{}

\author{Michael Kerber}{Institute of Geometry, Graz University of Technology, Austria \and \url{https://www.geometrie.tugraz.at/kerber/}}{kerber@tugraz.at}{https://orcid.org/0000-0002-8030-9299}{}

\author{Martin Uray}{Josef Ressel Centre for Intelligent and Secure Industrial Automation,\\ Salzburg University of Applied Sciences, Austria \and \url{https://martinuray.github.io/}}{martin.uray@fh-salzburg.ac.at}{https://orcid.org/0000-0002-8916-3847}{}

\authorrunning{S. Huber, K. Husz\'ar, M. Kerber, and M. Uray} %

\Copyright{Stefan Huber, Kristóf Huszár, Michael Kerber, and Marin Uray} %



\category{} %

\relatedversion{} %

\funding{The financial support by the Austrian Federal Ministry of Economy, Energy and Tourism, the National Foundation for Research, Technology and Development and the Christian Doppler Research Association is gratefully acknowledged.}%

\nolinenumbers %

\EventEditors{John Q. Open and Joan R. Access}
\EventNoEds{2}
\EventLongTitle{42nd Conference on Very Important Topics (CVIT 2016)}
\EventShortTitle{CVIT 2016}
\EventAcronym{CVIT}
\EventYear{2016}
\EventDate{December 24--27, 2016}
\EventLocation{Little Whinging, United Kingdom}
\EventLogo{}
\SeriesVolume{42}
\ArticleNo{23}
\usepackage{microtype} %
\usepackage{xcolor}
\usepackage{mathtools}
\usepackage{csquotes}
\usepackage{overpic}
\usepackage{cleveref}

\usepackage{physics}

\newcommand\nR{{\mathbb{R}}}
\DeclareMathOperator{\sgn}{sgn}

\newcommand{\score}{\mathrm{S}}
\newcommand{\eps}{\epsilon}
\renewcommand{\line}{\ell}
\newcommand{\kernel}{\kappa}

\newcommand{\defeq }{\vcentcolon=} %

\begin{document}

\maketitle

\begin{abstract}
We propose an alternative formulation of the well-known Hough transform to detect lines in point clouds. Replacing the discretized voting scheme of the classical Hough transform by a continuous score function, its persistent features in the sense of persistent homology give a set of candidate lines. We also devise and implement an algorithm to efficiently compute these candidate lines.
\end{abstract}

\section{Introduction}
\label{sec:introduction}

The \emph{Hough transform} is a classical method in Computer Vision to detect
lines (or more complicated geometric shapes) in partially-sampled and noisy
sceneries, dating back to the 1960s~\cite{mukhopadhyay2015}. Informally, lines are detected by a voting scheme where every
sample point $p$ votes for lines that contain $p$ or are close-by,
and those lines that accumulate many votes are considered to be part of
underlying scenery. To specify the meaning of ``close-by,''
the space of lines is parameterized by $\nR^2$ and discretized into pixels.
A sample point transforms into a planar curve in that line space and votes
for those lines whose pixel it traverses, see Figure~\ref{fig:hough_intro}.

The above voting scheme suffers from two major problems: first,
noise in the sample can lead to a collection of neighboring pixels that
all get large vote counts. If the selection rule is to chose the $k$ lines
with largest votes, this may result in a collection of lines that
are very close to each other, which might be not the desired outcome.
Second, the voting scheme based on the binary decision whether a pixel is traversed or not is unstable: even translating the grid (i.e., choosing a different
origin) may result in drastically different outcomes.

\subparagraph{Contributions.}
We propose a simple and efficient framework that addresses both of the
shortcomings. First, instead of voting on pixels, the sample points
continuously vote on all lines in line space by scoring each of them.
The score is determined by the Euclidean
distance of the point and the line and some kernel function, e.g., a linear or a Gaussian curve. This way, every candidate line receives a score from the
sample points, yielding a total score for any candidate line. The
resulting \emph{score function} is continuous over the parameter space of lines, and is stable under perturbations of the point sample, meaning that small changes in the sample will lead to controlled point-wise
changes in the score function.
The authors of~\cite{Princen1994} also follow a continuous treatment, but for a hypothesis testing framework.

The second idea concerns the selection of the most important lines based on the
score function. We suggest a natural approach to select the lines corresponding to the
local maxima of the score function as they received the highest score
in a neighborhood.
To sort these (potentially many) local maxima by importance,
we pick those maxima with the largest \emph{persistence} in the sense of
$0$-homology for the super-levelset filtration of the score
function~\cite{eh-computational}.\footnote{This idea of using persistence has
  already been proposed in a preliminary work~\cite{ferner2025}. However, it
  lacked the crucial idea of a continuous treatment via smooth voting
  functions, but relied on the original discretized setting, preventing
  a stability result as we now obtained in this work.}
Using persistence prevents that two close-by lines are both selected,
except when they are separated by a significant valley in the score function.
This addresses the problem of selecting two many similar lines.
A similar approach has been already applied to clustering~\cite{chazal2013}.
While the number of local maxima is not stable under perturbations,
the number of \emph{persistent} (or prominent) local maxima turns out to be stable, and their locations
also remain stable under certain conditions, see Theorem~\ref{thm:maxima-thm}.

We also describe an efficient algorithm that approximates the score function via a quad-tree subdivision. It guarantees that each local maximum with significant persistence has a representative line in the output, and that no line with insignificant persistence appears.

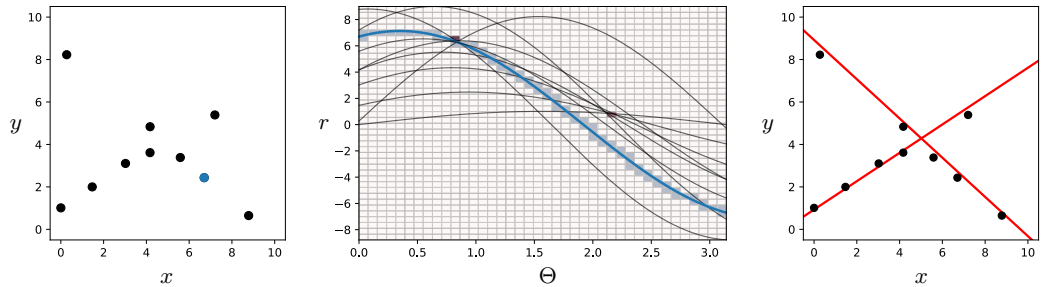
\begin{figure}[htb]
  \centering
  \begin{overpic}[width=.975\textwidth]{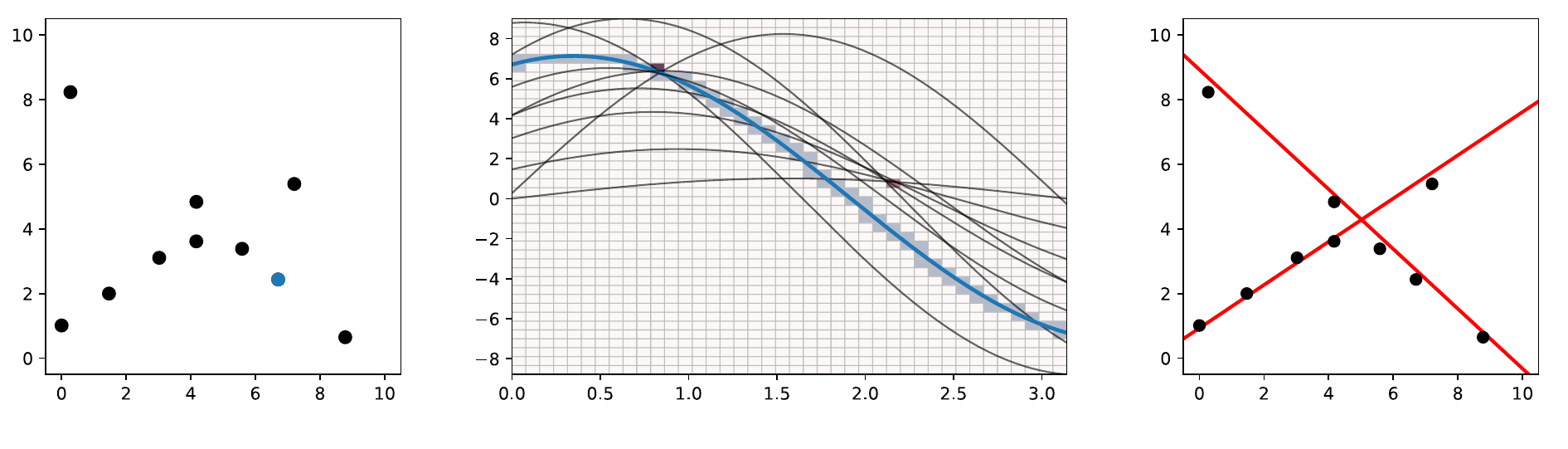}
  	\put (50,.5) {\footnotesize{$\Theta$}}
  	\put (13.5,.5) {\footnotesize{$x$}}
  	\put (86.25,.5) {\footnotesize{$x$}}
  	\put (-1,15.5) {\footnotesize{$y$}}
  	\put (28.75,15.5) {\footnotesize{$r$}}
  	\put (71.5,15.5) {\footnotesize{$y$}}
  \end{overpic}
  \caption{Left: Ten points sampled from two lines. Middle: The classical Hough transform passes to the space of lines in which every input point dualizes into
    a sinusoidal curve. The blue point on the left, for instance, is transformed
    to the blue curve in the middle. The pixels that the blue curve votes for
    are marked. In that example, there are two pixels (red) that receive five votes.
    Right: The centers of the high-count pixels are two
    lines in primal space that correspond to the two drawn lines.}
  \label{fig:hough_intro}
\end{figure}

\section{The score function}
\label{sec:score}

\subparagraph{Line parameterization.}
Set $M:=\nR\times [0,\pi]$.
For $(r,\Theta)\in M$, we define
\[
\line(r,\Theta):=\{(x,y)\in\nR^2\mid r=x\cos\Theta+y\sin\Theta\}.
\]
For $p=(r\cos\Theta,r\sin\Theta)$, $\line(r,\Theta)$ is the line through
$p$ that is perpendicular to the vector $\overrightarrow{p}$. Note that the set of all
lines in $\nR^2$ is in bijection with $M$, except that
$\line(r,0)=\line(-r,\pi)$, as the space of lines is homeomorphic to the open
M\"obius strip. For simplicity, we will first work with $M$ as a planar domain
and discuss the twisted glueing to a Möbius strip later.

\subparagraph{Score of a line.}
For a point $p$ and a line $\line$ in the Euclidean plane, write $\Delta(p,\line)$ for the
orthogonal distance of $p$ from $\line$.
We want that the score of $\line$ by $p$ equals $1$, if $\Delta(p,\line)=0$, and
otherwise it decreases as $\Delta(p,\line)$ grows. For that, we choose
a monotonously decreasing kernel function $\kernel:[0,\infty) \to [0,1]$
with $\kernel(0)=1$ and $\kernel(x)\to 0$
for $x\to\infty$. Two examples are
\[
  \kernel_\text{hat}(x) \defeq \max\left\{0,1-\frac{x}{\sigma} \right\}
  \qquad\text{and}\qquad
  \kernel_\text{RBF}(x) \defeq \exp\left(-\frac{x^2}{2\sigma^2}\right),\]
referred to as \emph{hat} and \emph{Gauss} kernel, respectively.
Now, for a finite set $P \subset \nR^2$ we define
\[
\score(p,\line)\defeq \kernel(\Delta(p,\line))
\qquad
\text{and}
\qquad
\score(\line)\defeq \score_P(\line)\defeq \frac{1}{|P|}\sum_{p\in P}\score(p,\line).
\]
We refer to $\score \colon M \to \nR, (r, \Theta)
\mapsto \score(\ell(r, \Theta))$ as the \emph{score function}.
Clearly, $0\leq\score(\line)\leq 1$. If $\kappa(x)=1$ only for $x=0$, then
$\score(\line) = 1$ implies that $\line$ contains all points of $P$.

The lines $\line(r,\Theta)$ containing a fixed $p=(x_0,y_0)$
are given by $r=x_0\cos\Theta + y_0\sin\Theta$,
whose solutions $(r,\Theta)$ form a sinusoidal curve $\gamma$, along which $\score(p,\ell(r, \Theta)) = 1$. Moreover, fixing $\Theta$ and varying $r$,
the distance of $p$ to the line $\line(r,\Theta)$ changes linearly with $r$.
Hence, the graph of $\score(p,\cdot)$ is obtained by sliding the function
graph of the kernel function $\kernel$ along $\gamma$. 
The graph of the score function $\score$ is obtained by
a normalized sum of $|P|$ such graphs.

\subparagraph{Stability.}
It is evident by construction that, if the kernel $\kernel$ is continuous/smooth,
then so is the score function $\score$. Moreover, the fact
that the Euclidean distance function to a fixed line is $1$-Lipschitz
immediately yields the stability of the score function.

\begin{lemma}
  Let $\kernel$ be Lipschitz-continuous with constant $\lambda$. Let $P'$
  denote an $\eps$-perturbation of $P$ (i.e., every point is replaced
  by a point in its $\eps$-ball). Then
  $
  \|\score_{P}-\score_{P'}\|_\infty\leq \lambda\eps.
  $
\end{lemma}

\section{Persistence-based selection}
For fixed $h_0\geq 0$, let
$ \score^{\geq h_0}:=\{\line\in M\mid \score(\line)\geq h_0\} $
denote the \emph{super-levelset} of $\score$.
By varying $h_0$ from $+\infty$ to $0$, we ``sweep  out'' the graph of $\score$
and obtain the super-levelset filtration. We consider the \emph{$0$-dimensional
persistent homology} of that filtration and sort the local maxima of $\score$
according to the persistence of their homology class.
For readers not familiar with persistent homology,
we offer an elementary explanation.
Assume for simplicity that all local maxima of $\score$ are isolated
and no two maxima have the same height.
Fix such a local maximum $x$ with height $h$. Call $x$ \emph{dominant} in
$\score^{\geq h'}$ if the connected component of $\score^{\geq h'}$ that
contains $x$ has $x$ as its highest point. Now, the persistence of $x$ is the
supremum over all $t\geq 0$ such that $x$ is dominant in $\score^{\geq
h-t}$.
We call $h$ the \emph{birth} level and $h-t$ the \emph{death} level of this component.
Lemma~\ref{lem:our-stability} is a reformulation of the stability result of persistence~\cite{ceh-stability}.

\begin{lemma}
  Let $\score_P, \score_{P'}$ be two score functions over $M$ with
  $\|\score_P-\score_{P'}\|_\infty=\eps$. Then, there exists a partial matching
  of local maxima of $\score_P$ to local maxima of $\score_{P'}$ such that all
  maxima of either score function with persistence greater than $\eps$ are
  matched, and the persistence of two matched local maxima differs by at most
  $2\eps$.
  \label{lem:our-stability}
\end{lemma}

\begin{theorem}
  \label{thm:maxima-thm}
  Let $\score_P$ be a score function with local maximum $x=(r,\Theta)$ and $h = \score_P(x)$
  of persistence $\alpha$.
  Let $\score_{P'}$ be a score function with $\|\score_P-\score_{P'}\|_\infty=\eps<\alpha/2$.
  Then, $\score_{P'}$ contains a local maximum $x'$ of persistence
  at least $\alpha-2\eps$
  such that $|\score_P(x')-h|\leq 2\eps$ and $x'$
  lies in the connected component of $x$ of the set $\score_P^{\geq h-\alpha}$.
\end{theorem}
\begin{proof}
By definition, $x$ is dominant in a connected component $C$ of $\score_P^{\geq h-\alpha}$.
Note that $\score_{P'}^{\geq h-\alpha+\eps}\subseteq \score_P^{\geq h-\alpha}$,
so $C\cap \score_{P'}^{\geq h-\alpha+\eps}$ is a connected component of $\score_{P'}^{\geq h-\alpha+\eps}$ (it is non-empty because of $x$). Its maximum $x'$ is easily seen to satisfy the three postulated properties.
\end{proof}

Theorem~\ref{thm:maxima-thm} tells us that for $\alpha\gg\eps$,
the persistent local maxima of $\score_{P'}$ approximate
those of $\score_P$ in the sense that for every persistent local maximum of $\score_P$,
we can find a persistent local maximum of $\score_{P'}$ with almost the same score
and being in the same super-levelset.

\section{Computation}\label{sec:computation}

We devise a simple approximation scheme of $\score_P$ based on a quad-tree
subdivision that yields a $\tilde{\score}_P$ with
$\|\score_P-\tilde{\score}_P\|_\infty\leq\eps$ for a fixed $\eps>0$ and is
constant on each quad-tree cell.\footnote{%
  The non-continuity of $\tilde{\score}_P$ does not cause a problem because
  $\tilde{\score}_P$ can be converted to a continuous score function with the
  same persistent homology. Without going into further detail, one triangulates
  each quad with a central Steiner vertex, considers a barycentric refinement
and then constructs a PL-function. We omit the details of this step and refer
to \cite[sec.~VI.3]{eh-computational}. }
  We first determine some $r_0\in\nR$ such that $\score_P \le \eps$ on
  $M\setminus [-r_0,r_0]\times[0,\pi]$, and set $\tilde{\score}_P$ zero there.

The main predicate for the subdivision part is a Lipschitz predicate that,
given a box $B:=[a,b]\times[c,d]\subset M$, yields a $\lambda \in \nR$ such that
$\score_P$, restricted to $B$, is continuous with Lipschitz constant $\lambda$.
Writing $h$ for the value of $\score_P$ at the midpoint of $B$, it follows that
the $\score_P$ at any point in $B$ differs from $h$ by at most $\lambda \cdot
\operatorname{diam}(B)/2$, where $\operatorname{diam}(B)$ denotes the diameter of $B$.
With that predicate, we simply subdivide the initial box
$[-r_0,r_0]\times[0,\pi]$ of the quad-tree and apply it until a subdivided box
satisfies $\lambda \cdot \mathrm{diam}(B)/2\leq\eps$. In that case, we set
$\tilde{\score}_P$ within that box $B$ to be the value of $\score_P$ at the
midpoint of $B$.

\begin{figure}[htb]
  \centering
  \includegraphics[height=5cm]{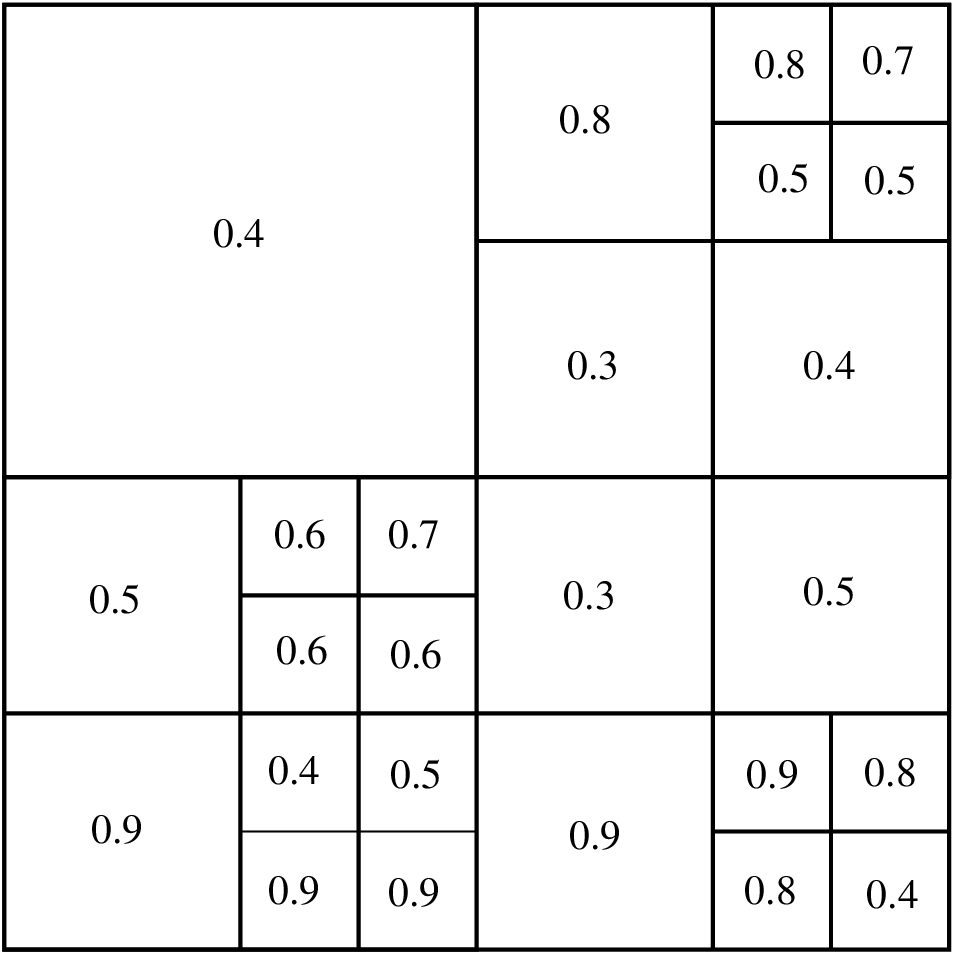}\\
  \vspace{0.6cm}
  
  \hspace{0.4cm}
  \includegraphics[height=3.5cm]{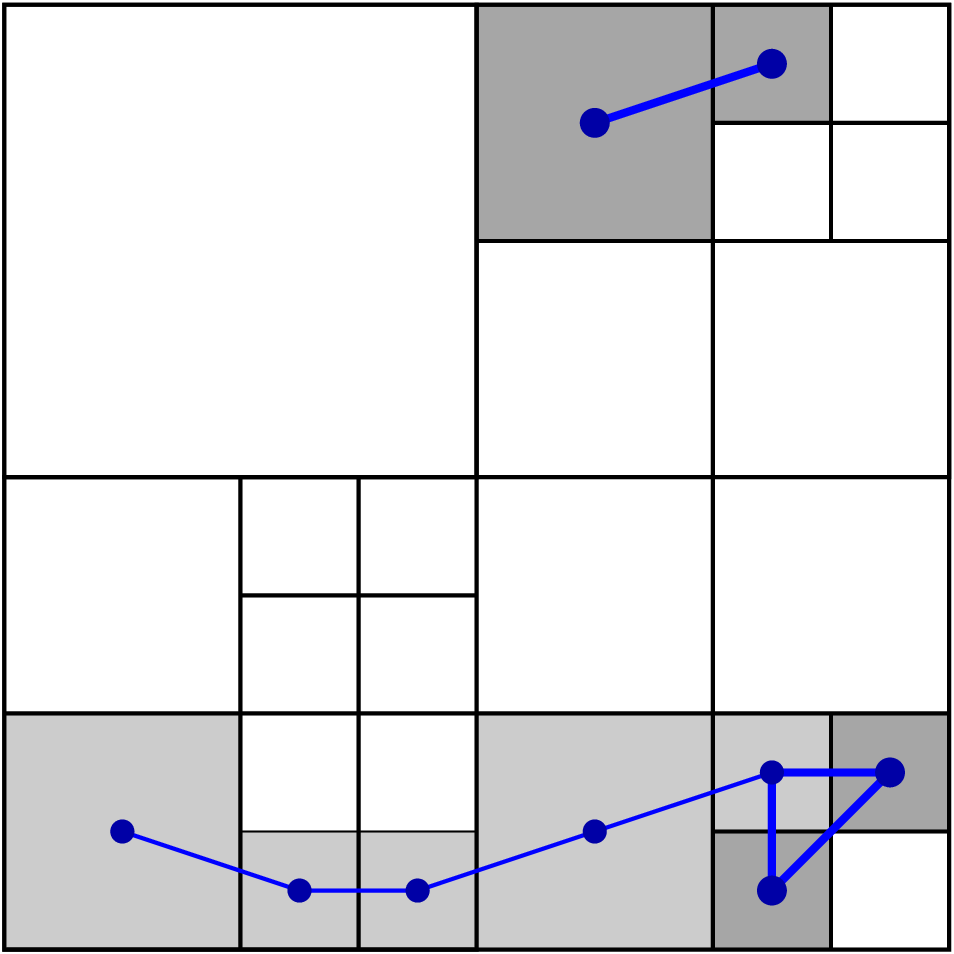}
  \hspace{0.3cm}
  \includegraphics[height=3.5cm]{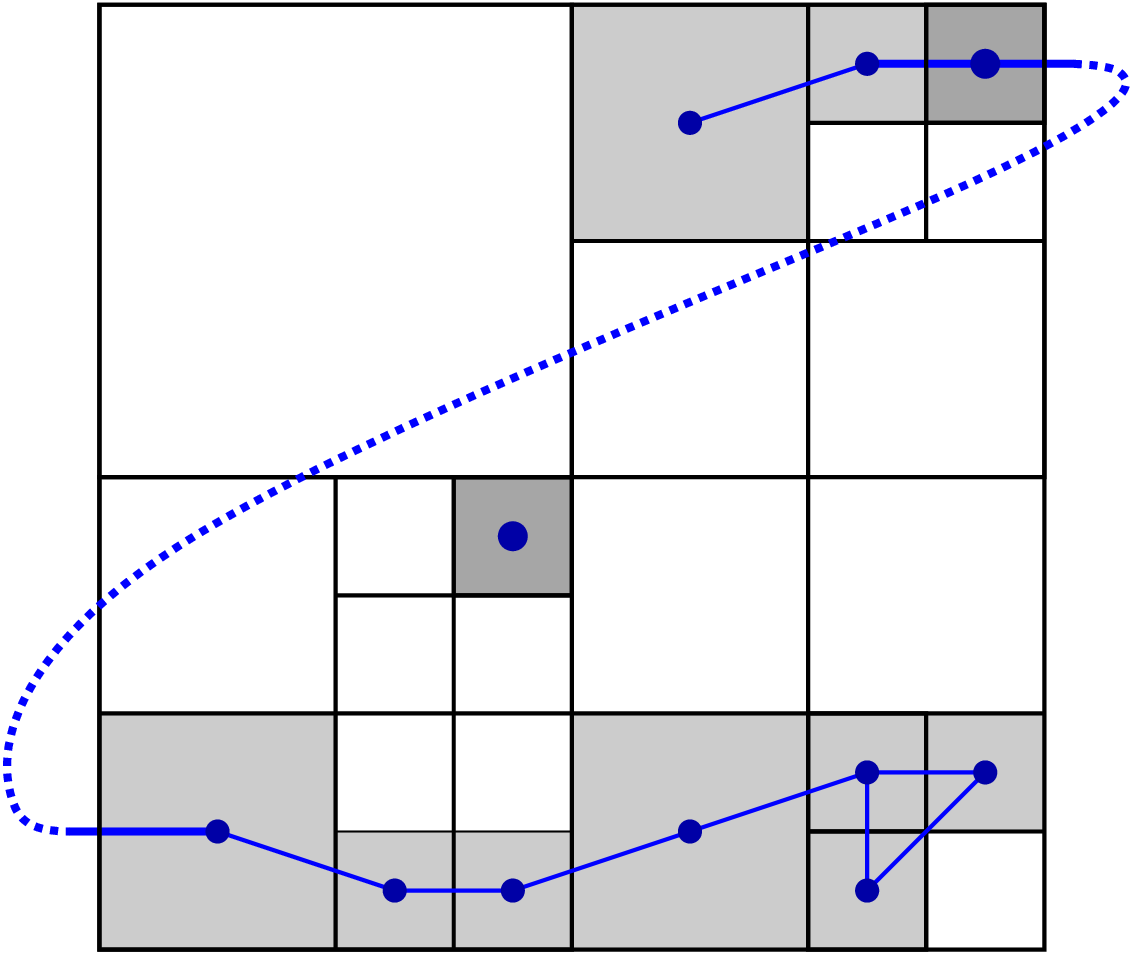}
  \hspace{0.3cm}
  \includegraphics[height=3.5cm]{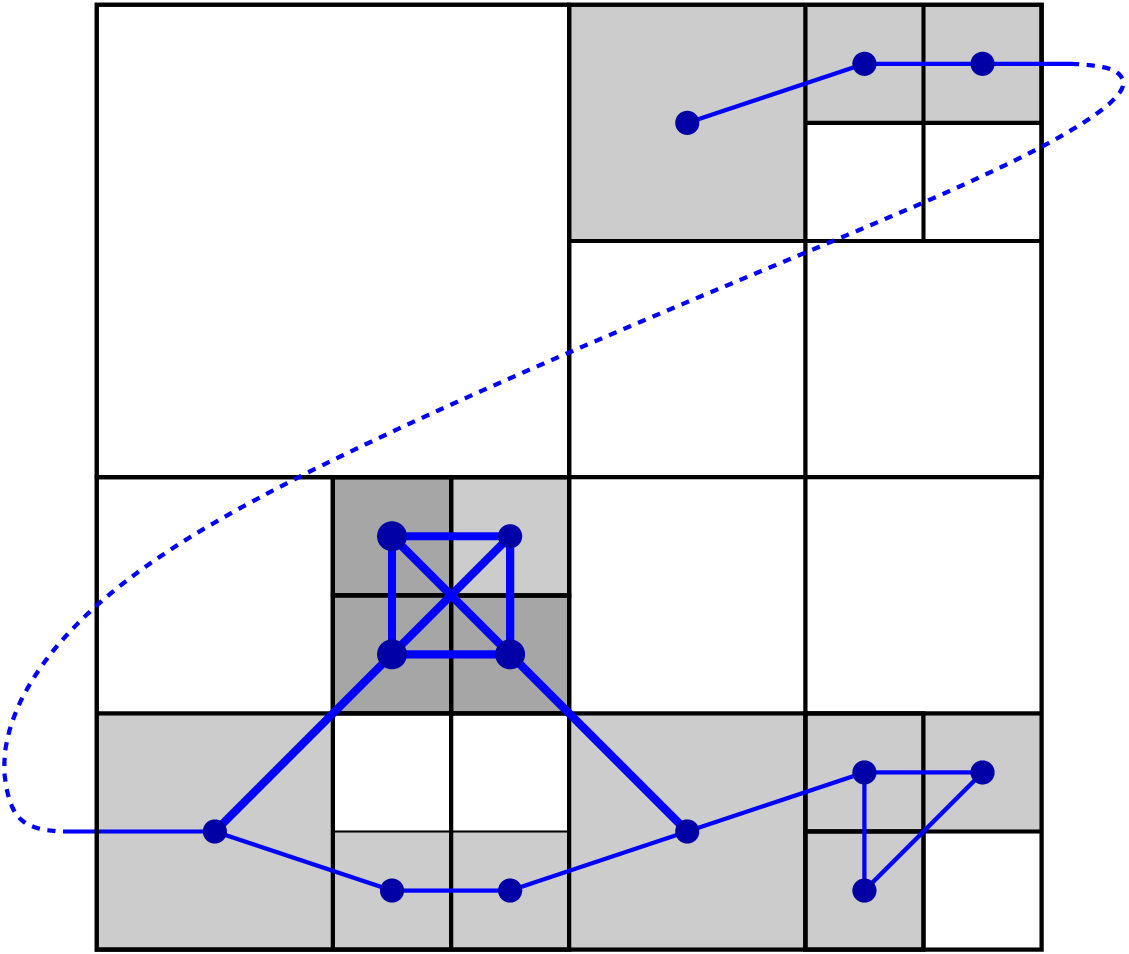}
  \caption{Top: the subdivided $\tilde{\score}_P$. Bottom: The super-levelsets for threshold $0.8$, $0.7$, and $0.6$,
    and the nerve of the boxes (the graph structure in blue). Vertices and edges that enter the graph at this
    threshold are drawn thicker. We observe that for $0.8$, a new connected component is formed in the upper region.
    For $0.7$, another isolated vertex is formed; moreover, the upper-right box connects with the lower-left box
    because of the M\"obius geometry, hence, the component formed at $0.8$ merges with another one, yielding
    a persistence of $0.1$ for the component. At $0.6$, the component formed at $0.7$ also gets merged, yielding
  a persistence of $0.1$ as well. Note that the graph is connected for $0.6$.}
  \label{fig:nerve}
\end{figure}

How to realize such a Lipschitz predicate?
Elementary calculus reveals that if the kernel function $\kernel$
has Lipschitz constant $\lambda_\kernel$ and all input points are in the
unit disk, $\score_P$ has a global Lipschitz constant
of $\lambda_\kernel\sqrt{2}$. %
However, using this global bound will
yield a uniform (and inefficient) subdivision. We obtain a better Lipschitz constant based on the locality
of the box. We outline a strategy for the two considered kernel functions
in Appendix~\ref{app:local_lipschitz}.

Computing the persistent homology (in dimension $0$) of $\tilde{\score}_P$ is done
via standard methods illustrated in Figure~\ref{fig:nerve}.
Via the Nerve Theorem~\cite{nerve}, the problem can be recast
into tracking the connected components of a graph that expands in terms of vertices
and edges, and that problem can be solved in almost linear time via
the union-find data structure~\cite{cormen}.
Note that in this step, we remember the twisted identification of the boundary
of $M$ (the domain of $\tilde{\score}_P$), effectively computing the persistent homology
over the M\"obius strip.

There is a choice how to ultimately select the local maxima
, and hence the lines in primal space, from the persistence information.
A simple choice is
to just select the $k$ local maxima with largest persistence, where $k$ is
a given constant. Another option is to fix a threshold $\alpha$ and pick
those local maxima with persistence at least $\alpha$. Here, $\alpha$ may be
chosen a-priori or a-posteriori~-- we demonstrate the latter in Section~\ref{sec:experiments}.

\section{A Use-case}
\label{sec:experiments}

We implemented a prototype of our method in Python, available at \url{https://github.com/JRC-ISIA/TopologicalHoughTransformation}.
Figure~\ref{fig:comparison_methods} shows an exemplary run of our method
on a point cloud sampled with noise from three lines (top-left).

\begin{figure}[htb]
    \centering
    \includegraphics[width=.8\textwidth]{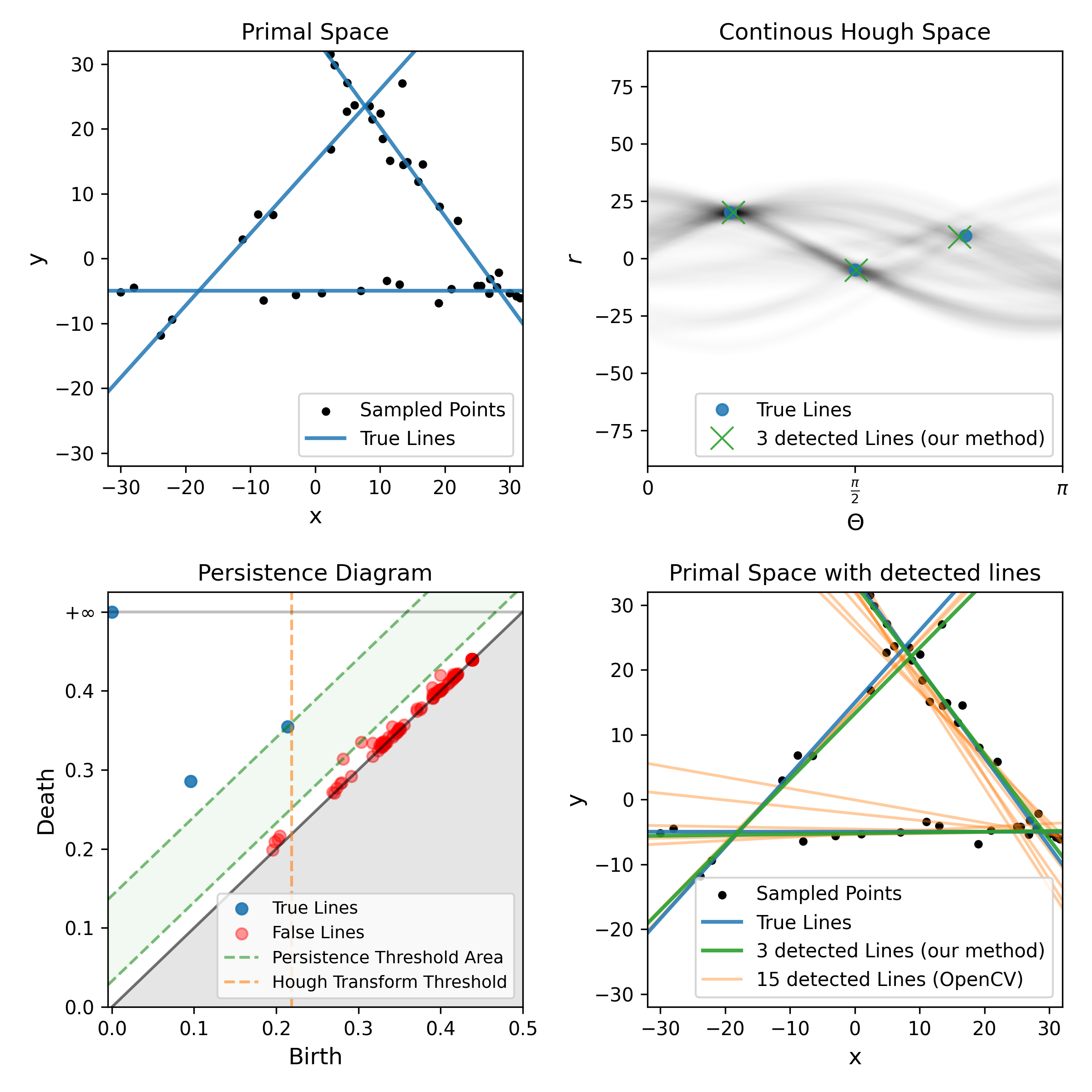}
    \caption{Illustration of the proposed method.}
    \label{fig:comparison_methods}
\end{figure}

Notably, each line
is sampled with a different number of points, resulting in different heights
for the three distinguished local maxima in the score surface. Nevertheless,
the score function has three distinguished ``bumps'' (top-right)
and the persistence homology reveals this fact: the \emph{persistence diagram}
(bottom-left) represents every local maximum by its height (birth) and the threshold
when it ceases to be dominant (death); its persistence is then the distance to the diagonal. With this interpretation, we observe three clearly distinguished
local maxima, and selecting them yields three lines that closely approximate
the ground truth (bottom-right).

In comparison, we display the outcome of the Hough transform as computed
by the \texttt{OpenCV} package~\cite{opencv_library}. In there, the specified threshold filters
local maxima by height, not by persistence. Consequently,
depending on how we choose the threshold, we either miss one of the
ground truth lines altogether, or the method returns
many local maxima (of low persistence) that all
approximate the line with the densest sampling because it creates a larger
bump in dual space. It would then require further post-processing to remove
these duplicates while this filtering step is elegantly integrated
in our framework already.

We ran more tests on instances similar to the one in Figure~\ref{fig:comparison_methods} and did some statistical analysis on the detected lines, showing
that the above situation is not cherry-picked but seems typical if
one deals with lines of different densities. See Appendix~\ref{app:experiments}
for details.

\newpage

\section{Discussion}
We have demonstrated that simple tools from computational geometry and topology
have the potential to increase the quality of line detection via the Hough
transform. 
Although the experiments presented are limited to line detection, the proposed 
method generalizes to arbitrary shapes by varying the shape parametrization and
the corresponding parameter space.
We are planning more extensive tests, including the use
of different kernel functions, evaluation on the effect of the kernel size, real image data, and the comparison with
state-of-art line detection methods.
A preliminary evaluation can be found in Appendix~\ref{app:experiments}.
As a prerequisite, we will first
implement a fast version of our prototype~-- we are optimistic that
the combination of adaptive subdivisions and the fact that persistence
in dimension $0$ can be computed in almost linear time will allow us
to apply our method to large image data sets.

\bibliography{biblio}

\appendix

\section{Gradients of the score function over the parameter space}
\label{app:global_lipschitz}

We consider the score function $(r, \Theta) \mapsto \score(r, \Theta)$ and
investigate a bound for the Lipschitz constant of this function. First, we
recall the definition
\begin{align}
  \score(r, \Theta) &= \frac{1}{|P|} \sum_{p \in P} \kappa(\Delta(p, \ell(r, \Theta))).
\end{align}

We are going to calculate $\nabla \score(r, \Theta) = (\pdv{\score}{r},
\pdv{\score}{\Theta})$ and derive a bound for the Lipschitz constant
$\lambda_\score$ of $\score$ via $\lambda_\score = \sup_{(r, \Theta) \in M}
\|\nabla \score(r, \Theta)\|$ with $\|.\|$ being the Euclidean norm.

\medskip

Note $\Delta(p, \ell(r, \Theta)) = |r - x_p \cos \Theta - y_p \sin \Theta|$
for $p = (x_p, y_p)$. Hence, we have
\begin{align}
  \left|\pdv{\Delta}{r}\right| &= |\sgn(r - x_p \cos \Theta - y_p \sin \Theta)| = 1, \\
  \left|\pdv{\Delta}{\Theta}\right| &= |\sgn(r - x_p \cos \Theta - y_p \sin \Theta) \cdot (x_p \sin \Theta - y_p \cos \Theta)| \le \|p\|.
\end{align}

Furthermore note that $|\pdv{\kappa}{\Delta}|$ is bound by the Lipschitz
constant $\lambda_\kappa$ of the kernel function~$\kappa$. We
therefore have
\begin{align}
  \left|\pdv{\score}{r}\right| &\le \frac{1}{|P|} \sum_{p \in P} \left|\pdv{\kappa}{\Delta}(r, \Theta)\right| \cdot \left|\pdv{\Delta}{r}\right| \le \lambda_\kappa \label{nabla1}\\
  \left|\pdv{\score}{\Theta}\right| &\le \frac{1}{|P|} \sum_{p \in P} \left|\pdv{\kappa}{\Delta}(r, \Theta)\right| \cdot \left|\pdv{\Delta}{\Theta}\right| \le \lambda_\kappa \cdot \frac{1}{|P|} \sum_{p \in P} \|p\| \le \lambda_\kappa \sup_{p \in P} \|p\|\label{nabla2}.
\end{align}

And this gives us the bound
\begin{align}
  \lambda_\score &= \sup_{(r, \Theta) \in M} \|\nabla \score(r, \Theta)\| \\
                 &\le \lambda_\kappa \sqrt{1 + \sup_{p \in P} \|p\|^2} \\
                 &\le \lambda_\kappa \sqrt{1 + D^2} \\
\end{align}
if $P$ is within a ball of radius $D$ around the origin. If we assume that all points
of $P$ are within the unit circle, we have $D=1$, hence a Lipschitz constant of $\lambda_\kappa \sqrt{2}$.

\section{Lipschitz constants in a box}
\label{app:local_lipschitz}
For a point $q=(x_q,y_q)$ and a (continuous) function $f:\nR\to\nR$, the \emph{vertical distance} is given by $|y_q-f(x_q)|$.
For an axis-aligned box $B$, the \emph{vertical distance of $B$ and $f$} is the minimum
of all vertical distances of $q\in B$ to $f$. In particular, that distance is $0$ if and only if the function
graph of $f$ enters the box.

Note that if $f$ is a sinusoidal curve consisting of all solutions $(r,\Theta)$ of the equation
$x_p \cos \Theta + y_p \sin \Theta=r$, the vertical distance can be computed
in constant time by considering the (up to two) local extrema
of $f$ within the $x$-range of $B$, and the corners of the $B$.

Using the vertical distance, we can derive an improved adaptive Lipschitz constant of the score function
restricted to a box $B$. For the kernel function $\kernel$ and $\delta\geq 0$,
we define
$\lambda_\kernel^{\geq \delta}$ as the Lipschitz constant of $\kernel$ on the interval $[\delta,+\infty)$.
By the properties of $\kernel$, $\lambda_\kernel^{\geq \delta}$ is non-increasing in $\delta$
and converges to $0$ for $\delta\to\infty$ for the two kernel functions
that we consider below.
The following lemma is a direct consequence of our definitions:

\begin{lemma}
  \label{lem:lll}
  For a fixed box $B$, let $\left.\score\right|_B(r,\Theta)$ denote the score function restricted to the box $B$.
  Let $P$ be a finite point set within the unit circle.
  For $p=(x_p,y_p)\in P$,
  let $\delta(p)$ denote the vertical distance of $B$ and the sinusoidal curve $r(\Theta) = x_p \cos \Theta + y_p \sin \Theta$.
  Then, $\left.\score\right|_B$ has a Lipschitz constant of at most
  \[\frac{\sqrt{2}}{|P|} \sum_{p \in P} \lambda_\kernel^{\geq \delta(p)}\]
\end{lemma}
\begin{proof}
  Note that the estimations in $(\ref{nabla1})$ and $(\ref{nabla2})$ from Appendix~\ref{app:global_lipschitz}
  remain valid (with $\|p\|\leq 1$ by assumption)
and moreover,
when restricting to $B$, $\kernel$ only can take values $\geq\delta(p)$ by definition, which yields
that $\left|\pdv{\kappa}{\Delta}\right|\geq \lambda_\kernel^{\geq \delta(p)}$.
\end{proof}

\subparagraph{The hat kernel}
Considering the definition of $\kernel:=\kernel_\text{hat}$, we observe at once that
\[
  \lambda_\kernel^{\geq \delta} = 
  \begin{cases}
    \frac{1}{\sigma} & \delta\leq \sigma \\
    0                & \delta\geq \sigma
  \end{cases}.
\]

Therefore, the computation of the local Lipschitz constant for $B$ boils down checking whether
the vertical distance for $B$ and the curve of $p\in P$ has vertical distance more than $\sigma$
or not. Writing $C$ for the number of curves with smaller vertical distance then $\sigma$,
we immediately get with Lemma ~\ref{lem:lll} a local Lipschitz constant of
\[\frac{C \sqrt{2}}{|P|\sigma}.\]
Note that $\lambda_\kernel=\frac{1}{\sigma}$, so with the trivial upper bound of $C\leq |P|$, the bound
becomes $\sqrt{2}\lambda_\kernel$, the global upper bound.

\subparagraph{The RBF kernel}
Set
\[ \kernel(x):=\kernel_\text{RBF}(x)=\exp\left(-\frac{x^2}{2\sigma^2}\right). \]
Basic calculus reveals that the derivative of $\kernel$ is
\[ \kernel'(x)=\frac{-x}{\sigma^2}\exp\left(-\frac{x^2}{2\sigma^2}\right) \]
whose absolute value is maximized
at $x=\sigma$ with absolute value equal to $\frac{1}{\sigma\sqrt{e}}$.
Moreover, the first derivative is decreasing in the interval $[\sigma,\infty)$.
Therefore, we can simply bound
\[
  \lambda_\kernel^{\geq \delta} = \begin{cases}
    \frac{1}{\sigma\sqrt{e}}                                             & \delta\leq \sigma \\
    \frac{\delta}{\sigma^2} \exp\left(-\frac{\delta^2}{2\sigma^2}\right) & \delta\geq \sigma
  \end{cases}.
\]
and obtain an adaptive bound using Lemma~\ref{lem:lll}.

\section{Experimental evaluation}
\label{app:experiments}
For the classical method (implemented in \texttt{OpenCV}~\cite{opencv_library}), the choice of a threshold is highly dependent on the input data (right in \Cref{fig:comparison_methods}),
making it difficult, if not impossible, to identify a single threshold that
enables the detection of the correct number of lines.
With the following experiment,
we demonstrate that the proposed
method is more likely to produce a threshold
that yields a clear and robust separation.
For this experiment, we generated $1000$ images, each containing four randomly
parameterized lines $(\Theta, r)$ with 18, 17, 16, and 15 sampled
points, respectively.
Each point was perturbed by an orthogonal displacement drawn uniformly from $\left[-1,1\right)$.
Our method is parametrized using the hat kernel $\kappa_\text{hat}$, $\sigma=5$, and $\eps=5$.
We measure the absolute gap between the last included and the first candidate line
excluded using persistence, $\Delta_{\text{Pers}}$ (ours), and vote counts,
$\Delta_{\text{Vote}}$ (\texttt{OpenCV}), taking the provided order of candidate lines into account.
The results are illustrated in \Cref{fig:exp1}.
While our method consistently yields a non-zero interval from which a valid
threshold can be selected, the vote-based approach exhibits
$\Delta_{\text{Vote}} = 0$ for $66.5\%$ of the experiments, thereby preventing the
selection of a threshold that correctly recovers the true number of lines.

\begin{figure}[htb]
    \centering
    \includegraphics[width=0.66\textwidth]{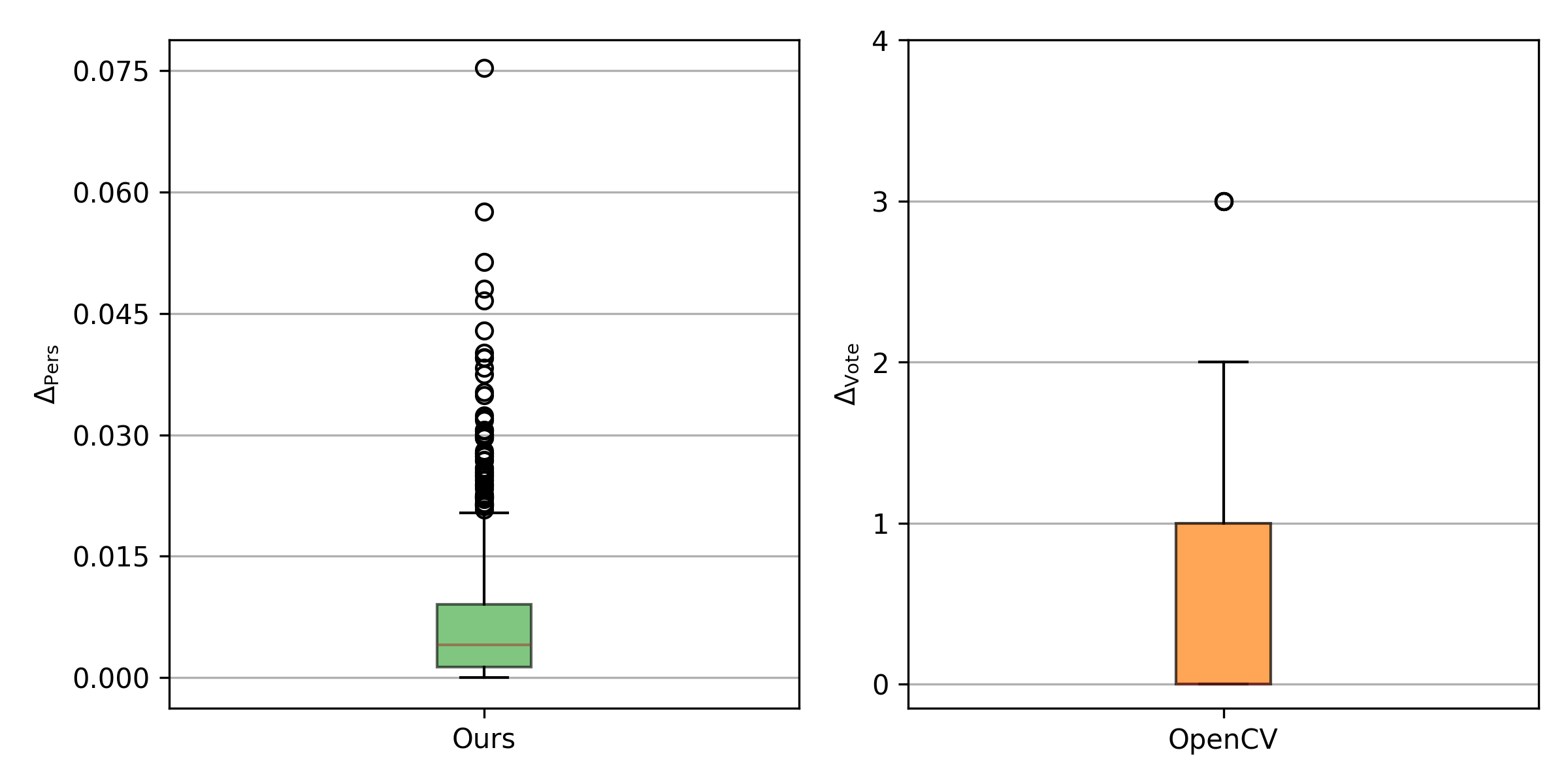}
    \caption{Illustration of the admissible threshold ranges
        $\Delta_\text{Vote}$ and $\Delta_\text{Pers}$ in Hough space for
        persistence-based thresholding (our method) and discrete, vote-based
        thresholding (\texttt{OpenCV}).}
    \label{fig:exp1}
\end{figure}

Further, assuming that both methods detect the optimal number of candidate lines in the
generated images, \Cref{fig:exp2} compares the quality of the detected lines
using the Euclidean distance in parameter space, and the absolute deviations in
$r$ and $\Theta$ to the ground truth.
Across all three evaluation metrics, the results consistently demonstrate a
performance advantage of our method over the baseline.

\begin{figure}[htb]
    \centering
    \includegraphics[width=1.0\textwidth]{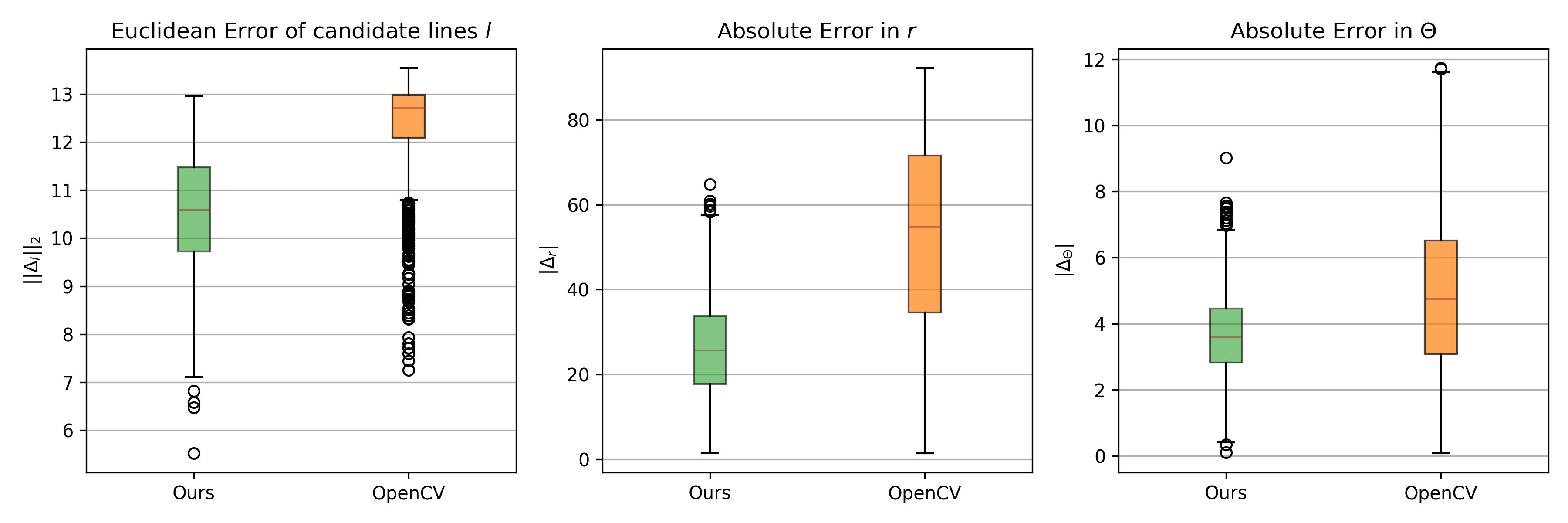}
    \caption{Results of the experiment evaluating detection quality using the
        Euclidean norm in normalized parameter space and the absolute errors in
        $r$ and $\Theta$ relative to the ground-truth lines, computed from the
        best pairwise matches.}
    \label{fig:exp2}
\end{figure}

To develop intuition on the optimal kernel width $\sigma$, we conducted a
preliminary experiment on a single random line detection across varying
noise levels, using the \emph{hat} kernel $\kappa_\text{hat}$ and a fix
$\eps=5$.
For each $\sigma \in [1,20]$, we performed 20 trials with different lines
randomly parameterized lines $(\Theta, r)$ per noise level $n \in [0,20]$,
where each point on the line is subject to an orthogonal displacement drawn
uniformly from $[-n,n)$.
The mean Euclidean error $\| \cdot \|$ across all trials is then computed for each
configuration, and the $\sigma$ value yielding the minimum error is identified.
The results, shown in Figure~\ref{fig:ratio_lambda_experiment}, suggest that the
optimal $\sigma$ corresponds to the expected noise level, implying that when the
maximum noise level can be estimated in advance, this prior knowledge can be
leveraged to maximize performance.

\begin{figure}[htb]
    \centering
    \includegraphics[width=1.0\textwidth]{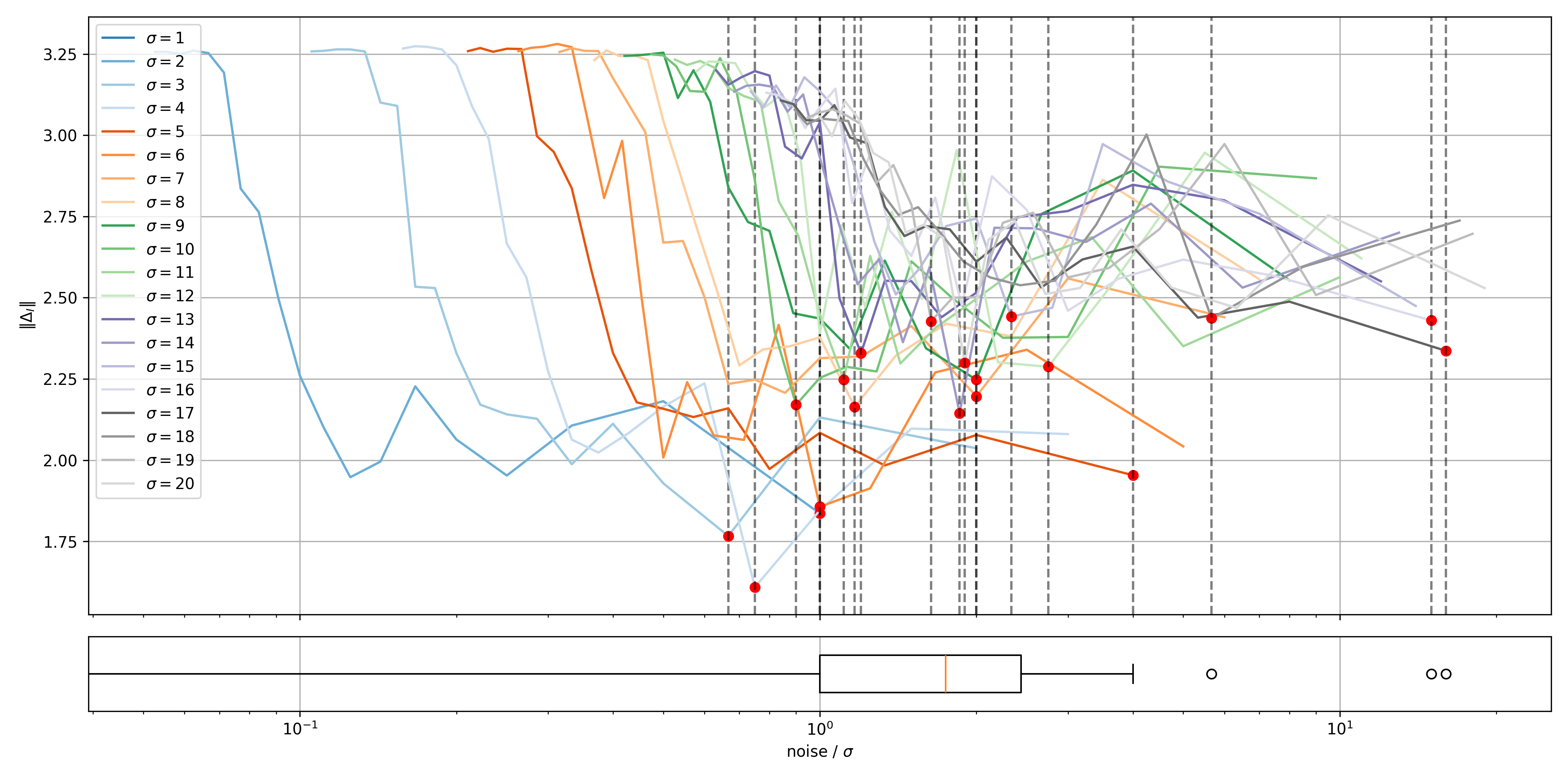}
    \caption{Analysis of the optimal kernel width $\sigma$. We plot the error
    as a function of the ratio between the noise level (in pixels) and
    $\sigma$. Each configuration is repeated 20 times and the mean error is
    reported. The minimum for each configuration is marked by a red dot and a
    vertical dashed line. The distribution of all minima is shown in the bottom
    boxplot.}
    \label{fig:ratio_lambda_experiment}
\end{figure}

To investigate the sensitivity of the method to the choice of $\eps$, we conduct
an experiment across varying values of $\eps$, fixing the \emph{hat} kernel
$\kappa_\text{hat}$ with $\sigma=5$.
For each configuration, $20$ trials are performed, each consisting of detecting
a single random line parameterized by $(\Theta, r)$, with $18$ sample points on a
$32\times 32$ image, where each point is subject to a random orthogonal displacement
drawn uniformly from $[-5, 5)$.
We report the Euclidean error $\| \cdot \|$ between the detected and the true
line as a sequence of boxplots over each value of $\eps$, alongside the per-run
detection runtime (in ms).

The results, shown in Figure~\ref{fig:eps_effect_experiment}, reveal a clear
trade-off between detection performance and computational cost as a function of $\eps$.
A lower $\eps$ yields finer quads and higher sensitivity, resulting in
better detection performance, but also a greater number of quads generated from
the continuous parameterization and thus a higher computational cost.
In this experimental setting, detection performance degrades as $\eps$ increases
up to $\eps=13$, beyond which both performance and runtime stabilize.
Notably, the computation speedup gained by relaxing detection performance
(e.g.\ $\eps=1$ vs.\ $\eps=13$) amounts to a factor of approximately $100$.

\begin{figure}[htb]
    \centering
    \includegraphics[width=1.0\textwidth]{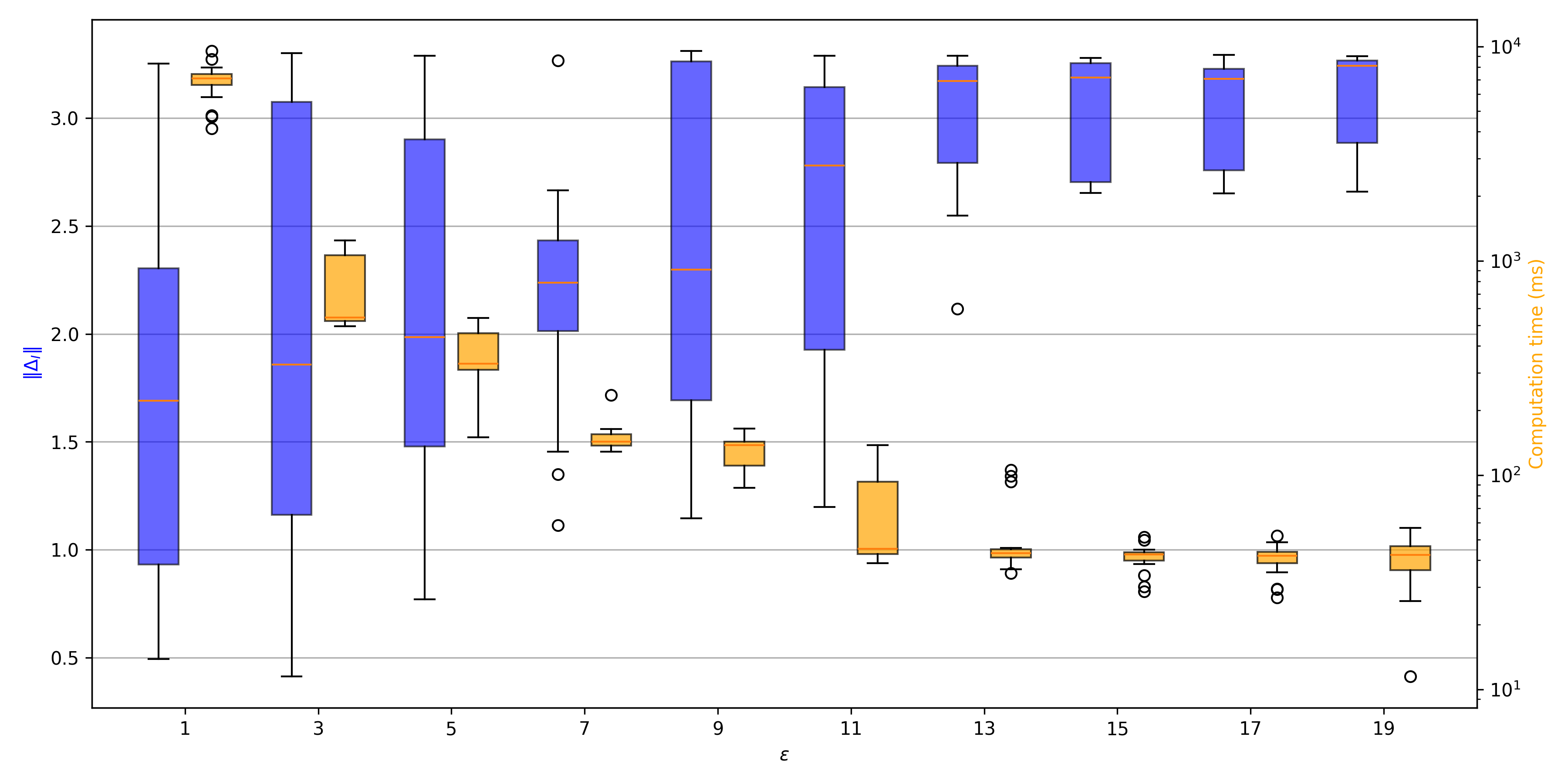}
    \caption{Detection performance and runtime (ms) as a function of $\eps$,
    evaluated over 20 trials per configuration on a single random line with
    18 sample points on a $32\times 32$ image.}
    \label{fig:eps_effect_experiment}
\end{figure}

\end{document}